\documentclass[journal,twoside,web]{ieeecolor}
\usepackage{lcsys}
\usepackage{cite}
\usepackage{amsmath,amssymb,amsfonts}
\usepackage{algorithmic}
\usepackage{graphicx}
\usepackage{textcomp}
\usepackage{xcolor}
\usepackage{cancel}
\usepackage[hidelinks]{hyperref}

\newcommand{\state}{x}
\newcommand{\sset}{\mathcal{X}}
\newcommand{\sdim}{n_{\state}}

\newcommand{\ctrl}{u}
\newcommand{\ctrlseq}{\mathbf{\ctrl}}
\newcommand{\cset}{\mathcal{U}}
\newcommand{\csetsafe}{\cset_{\text{s}}}
\newcommand{\controller}{\pi}
\newcommand{\cdim}{n_{\ctrl}}

\newcommand{\dstb}{d}

\newcommand{\ddim}{n_{\dstb}}

\newcommand{\tfunc}{l}

\newcommand{\safeset}{\mathcal{S}}

\newcommand{\vfunc}{V}

\newcommand{\vfsafe}{\vfunc_s}
\newcommand{\cost}{J}
\newcommand{\costsafety}{J_s}

\newcommand{\traj}{\xi}

\newcommand{\tvar}{t}
\newcommand{\tend}{T}
\newcommand{\tdummy}{\tau}
\newcommand{\tinit}{0}
\newcommand{\thor}{[\tinit,\tend]}
\newcommand{\dyn}{f}

\newcommand{\nbhd}{\mathcal{N}}
\newcommand{\ball}{\mathcal{B}}
\newcommand{\real}{\mathbb{R}}
\newcommand{\interior}[1]{{#1}^{\mathrm{o}}}
\newtheorem{problem}{Problem}
\newtheorem{theorem}{Theorem}
\newtheorem{proposition}{Proposition}
\newtheorem{note}{Note}
\newtheorem{definition}{Definition}

\DeclareMathOperator*{\argmin}{argmin}

\def\BibTeX{{\rm B\kern-.05em{\sc i\kern-.025em b}\kern-.08em
    T\kern-.1667em\lower.7ex\hbox{E}\kern-.125emX}}
\begin{document}
\title{Cooptimizing Safety and Performance with a Control-Constrained Formulation}
\author{Hao Wang$^{1*}$, Adityaya Dhande$^{2*}$, and Somil Bansal$^{1,3}$, \IEEEmembership{Member, IEEE}
\thanks{Manuscript received September 10, 2024; revised November 16, 2024; accepted November 27, 2024. This work is supported in part by the NSF CAREER Program under award 2240163, the DARPA ANSR program, and the IUSSTF-Viterbi program.(Corresponding author: Hao Wang)}
\thanks{*The authors contributed equally to this work.}
\thanks{$^{1}$The authors are associated with the Ming Hsieh Department of Electrical and Computer Engineering, University of Southern California.
        {(email: \tt\small haowwang@usc.edu)}}%
\thanks{$^{2}$The author is associated with the Department of Electrical Engineering, Indian Institute of Technology Bombay.
        {(email: \tt\small adityaya@iitb.ac.in)}}%
\thanks{$^{3}$The author is associated with the Aeronautics and Astronautics Department, Stanford University.
        {(email: \tt\small somil@stanford.edu)}}%
}
\pagestyle{empty}

\maketitle
\thispagestyle{empty}

\begin{abstract}
Autonomous systems have witnessed a rapid increase in their capabilities, but it remains a challenge for them to perform tasks both effectively and safely. The fact that performance and safety can sometimes be competing objectives renders the cooptimization between them difficult. One school of thought is to treat this cooptimization as a constrained optimal control problem with a performance-oriented objective function and safety as a constraint. However, solving this constrained optimal control problem for general nonlinear systems remains challenging. In this work, we use the general framework of constrained optimal control, but given the safety state constraint, we convert it into an equivalent \emph{control constraint}, resulting in a \emph{state and time-dependent control-constrained optimal control problem}. This equivalent optimal control problem can readily be solved using the dynamic programming principle. We show the corresponding value function is a viscosity solution of a certain Hamilton-Jacobi-Bellman Partial Differential Equation (HJB-PDE). Furthermore, we demonstrate the effectiveness of our method with a two-dimensional case study, and the experiment shows that the controller synthesized using our method consistently outperforms the baselines, both in safety and performance. The implementation of the case study can be found on the project website \footnote{\url{https://github.com/haowwang/cooptimize_safety_performance}}.
\end{abstract}

\begin{IEEEkeywords}
Autonomous systems, Optimal control, Robotics
\end{IEEEkeywords}

\section{Introduction}

\IEEEPARstart{P}ERFORMANCE and safety are two crucial factors we must consider when designing algorithms for autonomous systems. Clearly, we would like the systems to be effective in performing useful tasks. At the same time, the systems must satisfy safety requirements so that they do not inflict damage or injury. As a result, these two factors must be considered simultaneously when we are designing control algorithms. 

From the optimal control point of view, the existing methods can roughly be divided into two categories based on whether the safety requirement is posed as a constraint or objective in the optimization problem. Semantically, the latter means that safe behaviors are encouraged but not enforced. A large number of data-driven techniques \cite{bharadhwaj2020_conservative_safety_critic, srinivasan2020_safe_critic} fall into this category. One drawback of these techniques is that they do not provide any safety guarantees. 

The methods that treat the safety requirement as a constraint can be subdivided into two categories based on whether the safety requirement is considered simultaneously with the performance objective. One popular family of methods is safety filtering \cite{ames17cbfqp, Wabersich23filtering, hsu23filtering, borquez23filtering}, which provides safety-preserving interventions when necessary in runtime. They generally lead to myopic and suboptimal behaviors as, by design, the safety requirement is often not considered during the performance controller synthesis. 

On the other hand, one could formulate the problem as a state-constrained optimal control problem. With this framework, we can optimize the performance objective within the confines of the safety requirement and synthesize controllers that cooptimize safety and performance. However, state-constrained optimal control problems are notoriously difficult to solve using the dynamic programming principle unless certain controllability assumptions are satisfied \cite{soner1986_state_constrained_oc, capuzzo1990_hj_state_cons}. Alternatively, Model Predictive Control (MPC) techniques \cite{garcia1989_mpc, mayne2000_state_cons_mpc} have also been used to solve this problem. However, it is difficult to achieve optimality when the underlying problem involves nonlinear dynamics and/or non-convex state constraints. Recently, the authors in \cite{altarovici13} proposed a new framework to circumvent the controllability assumptions by characterizing the epigraph of the value function of the state-constrained optimal control problem. Though theoretically attractive, the method increases the dimensionality of the underlying optimal control problem and, in practice, is susceptible to several numerical challenges, as we demonstrate later in this paper.

In this work, we pose the problem of cooptimizing safety and performance as a state-constrained optimal control problem. Our key idea for overcoming the aforementioned challenges associated with state-constrained optimal control problems is to convert the state constraint into a control constraint using Hamilton-Jacobi reachability analysis. This results in an equivalent optimal control problem free of state constraints and can readily be solved using dynamic programming. We prove that the corresponding value function is a viscosity solution of a certain HJB-PDE, and it can be computed using existing Level Set methods and packages.

To summarize, the contribution of this letter is two-fold: 1) we propose a systematic way of converting a state-constrained optimal control problem into a \emph{control-constrained} optimal control problem and prove that two problems are equivalent, and 2) we show that the value function of the control-constrained optimal control problem is a viscosity solution to a final-value problem for a certain HJB-PDE. 

\section{Problem Formulation}\label{sec:prob_formulation}
In this work, we are interested in synthesizing controllers that optimize performance objectives for the given system while respecting the imposed safety constraint. We consider deterministic, continuous-time, and control-affine systems, governed by the ordinary differential equation $\frac{d\state}{d\tvar} = \dyn(\state,\ctrl) = \dyn_1(\state) + \dyn_2(\state)\ctrl$, where $\state\in\sset\subseteq\real^{\sdim}$ and $\ctrl\in\cset\subseteq\real^{\ddim}$ are the state and control of the system. We assume $\dyn$ is bounded and Lipschitz. We further assume the control space $\cset$ is convex. 

Let $r:\sset \times\cset \rightarrow \real$ and $\phi:\sset \rightarrow \real$ be the running cost over finite time horizon $[0,T)$ and final cost encoding the performance objectives. We assume both $r(\state,\ctrl)$ and $\phi(\state)$ are bounded and Lipschitz, and we further assume that $r(\state, \ctrl)$ is convex in $\ctrl$. Furthermore, the safety constraint is given by $\tfunc(\state)\geq 0 \ \forall \state\in\sset$, where $\tfunc$ is Lipschitz but is not required to be convex. 

We formalize the problem of interest as a state-constrained optimal control problem in Prob. \ref{prob:state_constrained_opt_ctrl_prob}. Let us use $\traj_{\state,\tvar}^{\ctrlseq}:[\tvar,\tend]\rightarrow\sset$ to denote the state trajectory starting from state $\state$ at time $\tvar$ evolved with control signal $\ctrlseq:[\tvar,\tend)\rightarrow \cset$. With a slight abuse of the notation, we use $\traj_{\state,\tvar}^{\ctrlseq}(\tdummy)$ to denote the state at time $\tdummy\geq\tvar$ along the trajectory $\traj_{\state,\tvar}^{\ctrlseq}$.
\begin{problem}[State-Constrained Optimal Control Problem]\label{prob:state_constrained_opt_ctrl_prob}
\begin{subequations}
\begin{align}
    \begin{split}\label{eq:state_constrained_oc_cost}
    &\inf_{\ctrlseq} \quad \cost(\state, \tvar,\ctrlseq) = \int_\tvar^\tend r(\traj_{\state,\tvar}^{\ctrlseq}(\tdummy), \ctrlseq(\tdummy))  d\tdummy \\ 
    & \qquad \qquad \qquad \qquad \qquad \ \ \ + \phi(\traj_{\state,\tvar}^{\ctrlseq}(\tend)) \\
    \end{split}\\
    &s.t.  \quad \frac{d}{d\tdummy}\traj_{\state,\tvar}^{\ctrlseq}(\tdummy) = \dyn(\traj_{\state,\tvar}^{\ctrlseq}(\tdummy), \ctrlseq(\tdummy))  \ \forall \tdummy \in [\tvar,\tend) \\
    & \qquad \ \tfunc(\traj_{\state,\tvar}^{\ctrlseq}(\tdummy))\geq 0  \ \forall \tdummy \in [\tvar,\tend] \label{eq:state_constraint}
    \end{align}
\end{subequations}
\end{problem}

Our goal in this work is finding the state-feedback controller $\controller^*:\sset\times[\tvar,\tend) \rightarrow \cset$ that solves Prob. \ref{prob:state_constrained_opt_ctrl_prob} at each state $\state\in\sset$ and time $\tvar\in\thor$. Since solving Prob. \ref{prob:state_constrained_opt_ctrl_prob} is challenging, we will present an equivalent optimal control problem whose solution is $\controller^*$.

\section{Background: Hamilton-Jacobi Reachability Analysis}

In this section, we provide a brief overview of Hamilton-Jacobi (HJ) reachability analysis, an approach we use to convert the state constraint \eqref{eq:state_constraint} into a control constraint. Given a state constraint \eqref{eq:state_constraint}, we use HJ reachability to determine the \emph{safe set} $\safeset$, the set of state $\state$ and time $\tvar$ starting from which the system can satisfy \eqref{eq:state_constraint} over time horizon $[\tvar,\tend]$. The construction of $\safeset$ is formulated as a minimum cost optimal control problem \cite{lygeros04,fialho94vi} with the cost functional $ \costsafety(\state, \tvar, \ctrlseq) =  \min_{\tdummy \in [\tvar,\tend]} \tfunc(\traj_{\state,\tvar}^{\ctrlseq}(\tdummy))$. The \emph{safety value function} at state $\state$ and time $\tvar$ is defined as 
\begin{equation}\label{eq:safety_value_function}
    \vfsafe(\state, \tvar) = \sup_{\ctrlseq} \costsafety(\state, \tvar, \ctrlseq) = \sup_{\ctrlseq} \min_{\tdummy \in [\tvar,\tend]} \tfunc(\traj_{\state,\tvar}^{\ctrlseq}(\tdummy))
\end{equation}
Then, the safe set $\safeset$ can be characterized using $\vfsafe(\state,\tvar)$ as $\safeset = \{(\state,\tvar) \in \sset\times\thor| \vfsafe(\state,\tvar) \geq 0\}$.

HJ reachability analysis provides a tractable means to compute the safety value function $\vfsafe(\state,\tvar)$. It has been shown that $\vfsafe(\state,\tvar)$ is the viscosity solution of the Hamilton-Jacobi-Bellman Variational Inequality (HJB-VI) \cite{lygeros04,fialho94vi}:
\begin{equation} \label{eq:HJBVI}
    \begin{aligned}
    &\min \biggl\{\frac{\partial\vfsafe}{\partial \tvar} + \max_{\ctrl\in\cset} \{\frac{\partial\vfsafe}{\partial \state}^\top \dyn(\state,\ctrl)\}, \tfunc(\state) - \vfsafe(\state, \tvar) \biggr\} = 0\\
    & \forall \state\in\sset \ \text{and} \ \forall\tvar\in[\tinit,\tend),  \, \vfsafe(\state, \tend) = \tfunc(\state) \ \forall \state\in\sset
    \end{aligned}
\end{equation}

\section{Method}
At its core, our method converts the state-constrained optimal control problem (Prob. \ref{prob:state_constrained_opt_ctrl_prob}) into a \emph{state and time-dependent control-constrained} optimal control problem, by explicitly characterizing the set of controls that leads the system to satisfy the state constraint, referred to as the \emph{set of safe controls}, at each state $\state$ and time $\tvar$. We first formalize the notion of set of safe controls and use it to formulate the state and time-dependent control-constrained optimal control problem (Prob. \ref{prob:ctrl_constrained_opt_ctrl_prob}). We then show Prob. \ref{prob:ctrl_constrained_opt_ctrl_prob} is equivalent to Prob. \ref{prob:state_constrained_opt_ctrl_prob}. Subsequently, we show the value function of Prob. \ref{prob:ctrl_constrained_opt_ctrl_prob} is a viscosity solution of a final-value problem for a certain HJB-PDE. Finally, we show one specific way of constructing the set of safe controls using HJ reachability analysis. 

\vspace{-0.5em}
\subsection{State and Time-Dependent Control-Constrained Optimal Control Problem}

We first provide the definition of the set of safe controls, inspired by a similar notion in \cite{borquez23filtering}.

\begin{definition}[Set of Safe Controls]\label{def:set_of_safe_controls}
    The set of safe controls at state $\state$ and time $\tvar$, denoted by $\csetsafe(\state,\tvar)$, is the set of controls that can instantaneously keep the system within the safe set $\safeset$. More precisely, 
    \vspace{-0.5em}
    \begin{equation}
         \csetsafe(\state,\tvar) = \{\ctrl\in\cset | \lim_{\epsilon\rightarrow 0} \vfsafe(\traj_{\state,\tvar}^\ctrl(\tvar+\epsilon), \tvar+\epsilon) \geq 0\}
    \end{equation}
    A set of safe controls is maximal if it contains all other sets of safe control, and we denote the maximal set of safe control by $\csetsafe^*$. 
\end{definition}

\begin{note}
    $\csetsafe(\state,\tvar)$ and $\csetsafe^*(\state,\tvar)$ can both be seen as set-value maps from $\sset\times[\tinit,\tend)$ to $\cset$.
\end{note}

The state and time-dependent control-constrained optimal control problem is presented below in Prob. \ref{prob:ctrl_constrained_opt_ctrl_prob}. It is worthwhile to note that Prob. \ref{prob:ctrl_constrained_opt_ctrl_prob} is identical to Prob. \ref{prob:state_constrained_opt_ctrl_prob}, only with the state constraint \eqref{eq:state_constraint} replaced by the control constraint \eqref{eq:ctrl_constraint} in Prob. \ref{prob:ctrl_constrained_opt_ctrl_prob}. We will also show that the optimal value of Prob. \ref{prob:ctrl_constrained_opt_ctrl_prob} is identical to that of Prob. \ref{prob:state_constrained_opt_ctrl_prob} for any state $\state\in\sset$ and time $\tvar\in\thor$, in Theorem. \ref{thm:problems_equivalence}. 

\begin{problem}[Control-Constrained Optimal Control Problem] 
\label{prob:ctrl_constrained_opt_ctrl_prob}
\begin{subequations}
    \begin{align}
    \begin{split}
        &\inf_{\ctrlseq}  \quad \cost(\state, \tvar,\ctrlseq) = \int_{\tvar}^\tend r(\traj_{\state,\tvar}^{\ctrlseq}(\tdummy), \ctrlseq(\tdummy)) d\tdummy \\ 
         &\qquad \qquad \qquad \qquad \qquad \qquad +\phi(\traj_{\state,\tvar}^{\ctrlseq}(\tend))  \\
    \end{split} \\
    & s.t.   \quad \frac{d}{d\tdummy}\traj_{\state,\tvar}^{\ctrlseq}(\tdummy) = \dyn(\traj_{\state,\tvar}^{\ctrlseq}(\tdummy), \ctrlseq(\tdummy))  \ \forall \tdummy \in [\tvar,\tend] \\
    & \quad \quad \ \ctrlseq(\tdummy) \in \csetsafe^*(\traj_{\state,\tvar}^{\ctrlseq}(\tdummy), \tau) \ \forall \tdummy \in [\tvar,\tend] \label{eq:ctrl_constraint}
    \end{align}
\end{subequations}
\end{problem}

\begin{theorem}\label{thm:problems_equivalence}
Let us denote the optimal value of Prob. \ref{prob:state_constrained_opt_ctrl_prob} and Prob. \ref{prob:ctrl_constrained_opt_ctrl_prob}, at state $\state\in\sset$ and time $\tvar\in\thor$, by $\vfunc_1(\state,\tvar)$ and $\vfunc(\state,\tvar)$. Then $\vfunc_1(\state,\tvar) = \vfunc(\state,\tvar) \ \forall \state\in\sset \ \text{and} \ \forall\tvar\in\thor$. 

\end{theorem}

\begin{proof}
    Take an initial state $\state\in\sset$ and initial time $\tvar\in\thor$. Let us denote the solutions to Prob. \ref{prob:state_constrained_opt_ctrl_prob} and Prob. \ref{prob:ctrl_constrained_opt_ctrl_prob}, from $\state$ and $\tvar$, by $\ctrlseq_1^*$ and $\ctrlseq^*$, respectively.  
    
    The system never violates the state constraint \eqref{eq:state_constraint} if $\ctrlseq^*$ is applied over $[\tvar,\tend)$, because $\ctrlseq^*(\tau)$ keeps the system within the safe set $\safeset$ instantaneously for all time $\tau\in[\tvar,\tend]$ by definition of the set of safe controls. With this fact established, we can now compare $\vfunc_1(\state,\tvar)$ and $\vfunc(\state,\tvar)$. 

    \textit{Case 1: $\tvar=\tend$}. In this case, $\vfunc_1(\state,\tend) = \vfunc(\state,\tend) = \phi(\state)$.
    
    \textit{Case 2: $\tvar\in[\tinit,\tend)$}. By definition of the state-constrained optimal control problem, we have $\vfunc_1(\state,\tvar) \leq \vfunc(\state,\tvar) \ \forall \state\in\sset \ \text{and} \ \forall\tvar\in[\tinit, \tend)$. 

    Now we would like to prove $\vfunc(\state,\tvar) \leq \vfunc_1(\state,\tvar) \ \forall \state\in\sset \ \text{and} \ \forall\tvar\in[\tinit, \tend)$. Before proceeding, we will establish the fact that $\ctrlseq^*_1$ satisfies the control constraint \eqref{eq:ctrl_constraint} for all $\tdummy\in[\tvar,\tend)$. Suppose that is not the case. Then $\exists \tdummy\in[\tvar,\tend)$ such that $\ctrlseq^*_1(\tdummy) \notin \csetsafe(\traj_{\state,\tvar}^{\ctrlseq_1^*}(\tdummy), \tdummy)$. As a result, $\lim_{\epsilon\rightarrow 0}\vfsafe(\traj_{\state,\tvar}^{\ctrlseq_1^*}(\tdummy+\epsilon), \tdummy+\epsilon) < 0$. By definition of the safety value function $\vfsafe(\state,\tvar)$ \eqref{eq:safety_value_function}, the system would certainly violate the state constraint at some point over the time horizon $[\tdummy+\epsilon,\tend]$. However, $\ctrlseq^*_1$ is a solution of Prob. \ref{prob:state_constrained_opt_ctrl_prob} and hence will not lead the system to violate the state constraint over the time horizon $[\tvar,\tend]$. We have reached a contradiction, and therefore $\ctrlseq^*_1(\tdummy)$ satisfies the control constraint \eqref{eq:ctrl_constraint} for all $\tdummy\in[\tvar,\tend)$.
    
    Take solution $\ctrlseq^*_1$ of Prob. \ref{prob:state_constrained_opt_ctrl_prob} at initial state $\state$ and time $\tvar$. Since $\ctrlseq^*_1(\tdummy)$ satisfies the control constraint \eqref{eq:ctrl_constraint} for all $\tdummy\in[\tvar,\tend)$, $\ctrlseq^*_1$ is feasible for Prob. \ref{prob:ctrl_constrained_opt_ctrl_prob}. Therefore, $\vfunc(\state,\tvar)\leq\vfunc_1(\state,\tvar)$. 

    Hence, we have shown that $\vfunc_1(\state,\tvar) = \vfunc(\state,\tvar)$ for any state $\state\in\sset$ and time $\tvar\in\thor$.
\end{proof}

\subsection{Solving Control-Constrained Optimal Control Problem}

For the remainder of this letter, we use $\vfunc(\state,\tvar)$ to denote the value function of Prob. \ref{prob:ctrl_constrained_opt_ctrl_prob}. 
We introduce the following result regarding $\vfunc(\state,\tvar)$, and the proof of this result is heavily inspired by the proof of Theorem 10.2 in \cite{evans2010pde}. 

\begin{theorem}
    Assume the set-valued map $\csetsafe^*:\sset\times\thor\rightrightarrows\cset$ is lower hemicontinuous. The value function $\vfunc(\state,\tvar)$ is a viscosity solution of the following final-value problem for the HJB-PDE 
    \begin{equation}\label{eq:hjb_pde_perf}
    \begin{aligned}
        &\frac{\partial \vfunc}{\partial \tvar} + \min_{\ctrl\in\csetsafe^*(\state,\tvar)}\{\dyn(\state,\ctrl)\top \frac{\partial\vfunc}{\partial\state} + r(\state,\ctrl)\} = 0 \\
        &\forall \state\in\sset \ \text{and} \  \forall \tvar\in [\tinit,\tend), \vfunc(\state,\tend) = \phi(\state) \ \forall \state\in\sset
    \end{aligned}
\end{equation}
\end{theorem}

\begin{proof}
    The continuity of $\vfunc(\state,\tvar)$ can be established using similar arguments presented in Lemma 10.3.3 in \cite{evans2010pde}. For brevity, we will not show $\vfunc(\state,\tvar)$ is continuous in this proof. We will first show that $\vfunc(\state,\tvar)$ is a viscosity supersolution.
    Take test function $\psi\in C^1(\sset\times\thor)$ and assume that $\vfunc - \psi$ has a local maximum at $(\state_0, \tvar_0)$. We must show that $\frac{\partial\psi}{\partial \tvar}|_{(\state_0,\tvar_0)} + \min_{\ctrl\in\csetsafe^*(\state_0,\tvar_0)} \{\frac{\partial\psi}{\partial\state}|_{(\state_0,\tvar_0)}^\top \dyn(\state_0,\ctrl_0) + r(\state_0,\ctrl)\}\geq 0$. 
    Suppose that is not the case. Then $\exists \ctrl_0\in\csetsafe^*(\state_0,\tvar_0) \ \text{and} \ \exists \theta>0$ such that 
    \begin{equation}\label{eq:viscosity_proof_supersolution_condition}
        \frac{\partial\psi}{\partial \tvar}|_{(\state_0,\tvar_0)} + \frac{\partial\psi}{\partial\state}|_{(\state_0,\tvar_0)}^\top \dyn(\state_0,\ctrl_0)+ r(\state_0,\ctrl_0)\leq-\theta<0
    \end{equation}
    
    Since $\dyn$ and $r$ are continuous in $\state$ and $\ctrl$, for $(\state,\ctrl,\tvar)$ that is sufficiently close to $(\state_0,\ctrl_0,\tvar_0)$, or equivalently 
    $||\state-\state_0||_2 + ||\ctrl-\ctrl_0||_2 + |\tvar-\tvar_0| < \delta$ for some $\delta>0$, condition \eqref{eq:viscosity_proof_supersolution_condition} holds.
    We denote the neighborhoods $||\state-\state_0||_2 + |\tvar-\tvar_0| < \frac{\delta}{2}$ and $||\ctrl - \ctrl_0||_2 <\frac{\delta}{2}$, by $\hat{\nbhd}$ and $\hat{\cset}$, respectively. 

    Take $\hat{\cset}$. $\hat{\cset} \cap \csetsafe^*(\state_0,\tvar_0)\neq\emptyset$ because $\ctrl_0\in\hat{\cset}$ and $\ctrl_0\in\csetsafe^*(\state_0,\tvar_0)$. Since by assumption $\csetsafe^*(\state,\tvar)$ is lower hemicontinuous, there exists a neighborhood $\nbhd$ of $(\state_0,\tvar_0) \ s.t. \ \forall (\state,\tvar)\in\nbhd, \csetsafe^*(\state,\tvar)\cap\hat{\cset}\neq\emptyset$. It follows immediately that $\forall (\state,\tvar)\in\hat{\nbhd}\cap\nbhd, \csetsafe^*(\state,\tvar)\cap\hat{\cset}\neq\emptyset$. Then by continuity of $\dyn$ in $\ctrl$, there exists $\tvar_e>\tvar_0$, over which we can construct a control signal $\ctrlseq^*:[\tvar_0,\tvar_e)\rightarrow\cset$, along with the resulting state trajectory $\traj_{\state_0,\tvar_0}^{\ctrlseq^*}:[\tvar_0,\tvar_e]\rightarrow\sset$, such that $(\traj_{\state_0,\tvar_0}^{\ctrlseq^*}(\tdummy), \tdummy)\in\hat{\nbhd}\cap\nbhd \ \forall \tdummy\in[\tvar_0,\tvar_e]$ and concurrently  $\ctrlseq^*(\tdummy)\in\csetsafe^*(\traj_{\state_0,\tvar_0}^{\ctrlseq^*}(\tdummy),\tdummy)\cap\hat{\cset} \ \forall \tdummy\in[\tvar_0,\tvar_e)$. By construction, $(\traj_{\state_0,\tvar_0}^{\ctrlseq^*}(\tdummy),\tdummy)\in\hat{\nbhd}$ and $\ctrlseq^*(\tdummy)\in\hat{\cset} \ \forall \tdummy\in(\tvar_0,\tvar_e)$, and hence $(\traj_{\state_0,\tvar_0}^{\ctrlseq^*}(\tdummy), \ctrlseq^*(\tdummy), \tdummy) \ \forall \tdummy\in(\tvar_0,\tvar_e)$ satisfies condition \eqref{eq:viscosity_proof_supersolution_condition}. 

    By assumption $\vfunc-\psi$ has a local maximum at $(\state_0,\tvar_0)$, we have $ \vfunc(\state, \tvar) - \psi(\state, \tvar) \leq \vfunc(\state_0, \tvar_0) - \psi(\state_0, \tvar_0) \ \forall (\state,\tvar)\in\hat{\nbhd}$. Note that from the dynamics programming principle we have $ \vfunc(\traj_{\state_0,\tvar_0}^{\ctrlseq}(\tvar_0),\tvar_0) \leq \int_{\tvar_0}^{\tvar_e} r(\traj_{\state_0,\tvar_0}^{\ctrlseq}(\tvar), \ctrlseq(\tvar)) d\tvar \  + \vfunc(\traj_{\state_0,\tvar_0}^{\ctrlseq}(\tvar_e), \tvar_e)$ for any control signal $\ctrlseq$ that satisfies the control constraint \eqref{eq:ctrl_constraint} over the time horizon $[\tvar_0, \tend)$. Making use of this fact and rearranging the equation we arrive at the following:
    \begin{equation*}
        \begin{aligned}
             &0 \leq \psi(\traj_{\state_0,\tvar_0}^{\ctrlseq^*}(\tvar_e),\tvar_e) - \psi(\traj_{\state_0,\tvar_0}^{\ctrlseq^*}(\tvar_0),\tvar_0) \\
             &\qquad \qquad - \vfunc(\traj_{\state_0,\tvar_0}^{\ctrlseq^*}(\tvar_e),\tvar_e)+ \vfunc(\traj_{\state_0,\tvar_0}^{\ctrlseq^*}(\tvar_0),\tvar_0) \\
            &\leq \psi(\traj_{\state_0,\tvar_0}^{\ctrlseq^*}(\tvar_e),\tvar_e) - \psi(\traj_{\state_0,\tvar_0}^{\ctrlseq^*}(\tvar_0),\tvar_0) \ \bcancel{- \vfunc(\traj_{\state_0,\tvar_0}^{\ctrlseq^*}(\tvar_e),\tvar_e)} \\
            &\qquad + \int_{\tvar_0}^{\tvar_e} r(\traj_{\state_0,\tvar_0}^{\ctrlseq^*}(\tvar), \ctrlseq^*(\tvar))d\tvar \  \bcancel{+ \vfunc(\traj_{\state_0,\tvar_0}^{\ctrlseq^*}(\tvar_e), \tvar_e)} \\
            &= \int_{\tvar_0}^{\tvar_e} \biggl[\frac{\partial}{\partial \tvar}\psi(\traj_{\state_0,\tvar_0}^{\ctrlseq^*}(\tvar),\tvar) + \frac{\partial}{\partial \state}\psi(\traj_{\state_0,\tvar_0}^{\ctrlseq^*}(\tvar),\tvar)^\top  \\
            & \qquad \qquad \qquad  \dyn(\traj_{\state_0,\tvar_0}^{\ctrlseq^*}(\tvar),\ctrlseq^*(\tvar)) + r(\traj_{\state_0,\tvar_0}^{\ctrlseq^*}(\tvar), \ctrlseq^*(\tvar)) \biggr] d\tvar \\
            &\leq -\theta(\tvar_e - \tvar_0)
        \end{aligned}
    \end{equation*}
    Since $\theta > 0$ and $(\tvar_e - \tvar_0)>0$, we have reached a contradiction. Therefore, we have $\frac{\partial\psi}{\partial \tvar}|_{(\state_0,\tvar_0)}+ \min_{\ctrl\in\csetsafe^*(\state_0,\tvar_0)}\{\frac{\partial\psi}{\partial\state}|_{(\state_0,\tvar_0)}^\top \dyn(\state_0,\ctrl)  + r(\state_0,\ctrl)\}\geq 0$, and $\vfunc$ is a viscosity supersolution. 

    The proof for the viscosity subsolution follows a similar argument, and we omit the proof for brevity. Because $\vfunc$ is both a viscosity supersolution and subsolution, it is a viscosity solution of the final-value problem for the HJB-PDE \eqref{eq:hjb_pde_perf} as we intended to show. 
\end{proof}

\subsection{Characterizing the Set of Safe Controls}\label{subsec: set_of_safe_controls}
We now introduce a method to characterize the set of safe controls using HJ reachability analysis. Given a state constraint \eqref{eq:state_constraint}, we first obtain the safety value function $\vfsafe$ \eqref{eq:safety_value_function} by solving the HJB-VI \eqref{eq:HJBVI}. Then, for $0<\gamma<<1$, we characterize the set of safe control at state $\state$ and time $\tvar$ in \eqref{eq:set_of_safe_controls_hj}, and we show that when interpreted as a set-value map, \eqref{eq:set_of_safe_controls_hj} is lower hemicontinuous. 

\begin{definition}[Set of Safe Controls Using HJ Reachability]
    \begin{equation} \label{eq:set_of_safe_controls_hj}
        \csetsafe(\state,\tvar)=\left\{\begin{array}{l}
        \cset \text { if } \vfsafe(\state, \tvar) > 0 \\ \\
        \{\ctrl\in\cset| -\gamma \leq \frac{\partial\vfsafe}{\partial \tvar} + \frac{\partial\vfsafe}{\partial \state}^\top \dyn(\state, \ctrl) \leq 0 \} \\  
        \quad \quad \quad \quad \quad \quad \quad \text { if } \vfsafe(\state, \tvar) \leq 0 \end{array}\right.
    \end{equation}
\end{definition}

\begin{note}
    \eqref{eq:set_of_safe_controls_hj} is only a set of safe control per Definition \ref{def:set_of_safe_controls} when $\gamma = 0$, since when $\gamma>0$, $\csetsafe(\state,\tvar)$ contains controls that would decrease the safety value and lead the system out of the safe set. In practice we do use $\gamma=0$, because the construct of $\gamma$ is only necessary to render $\csetsafe(\state,\tvar)$ lower hemicontinuous. There are other ways of constructing $\csetsafe(\state,\tvar)$ beyond \eqref{eq:set_of_safe_controls_hj}, but it is challenging to construct a $\csetsafe(\state,\tvar)$ satisfying all the following criteria: 1) safety-preserving, 2) maximal, and 3) lower-hemicontinuous. We consider \eqref{eq:set_of_safe_controls_hj} as the best option given the fact that it is maximal, safety-preserving in practice ($\gamma=0)$, and easy to explain and implement. 
\end{note}

\begin{figure}[tbh!]
    \centering
    \includegraphics[width=1\linewidth]{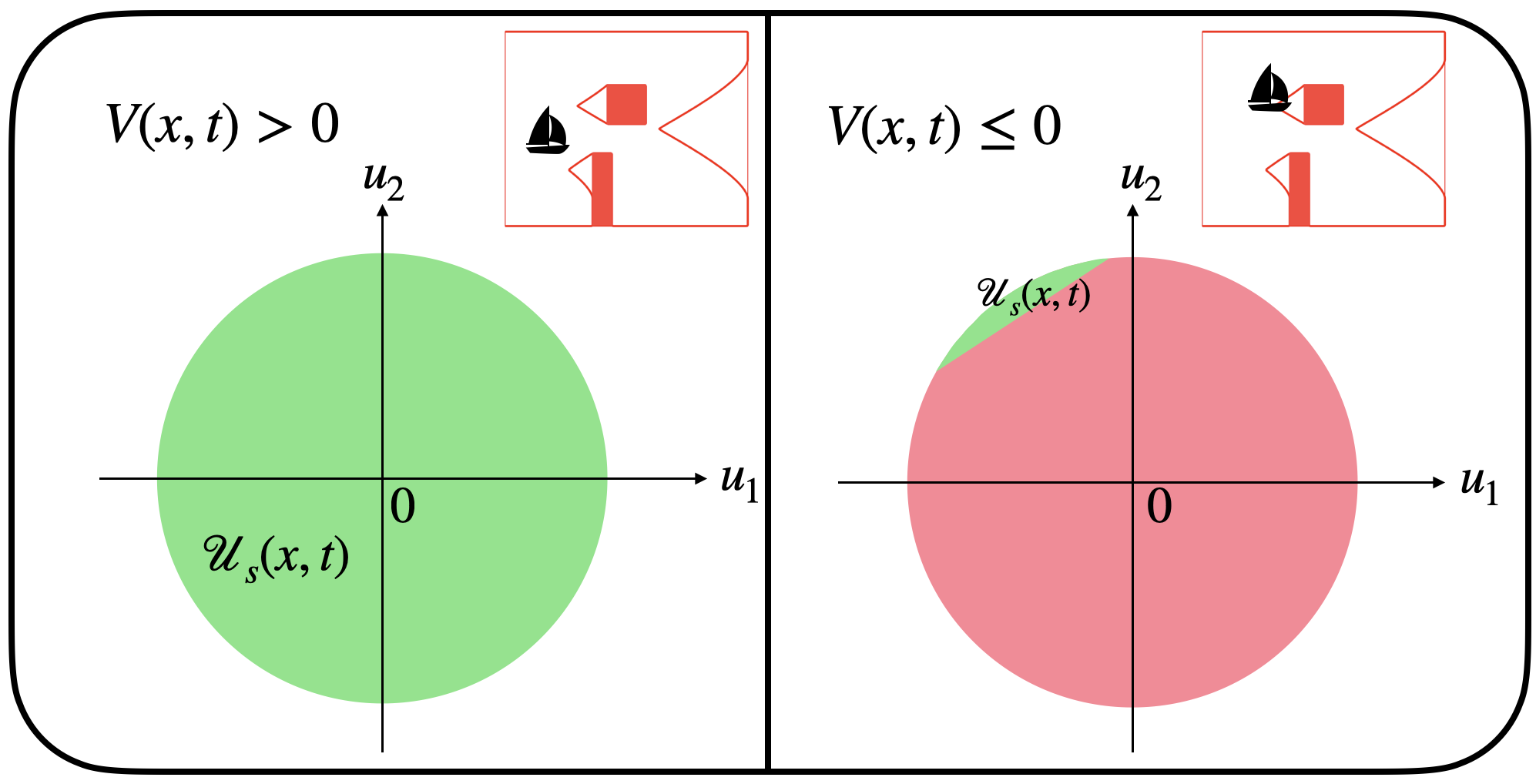}
    \caption{An illustration of the set of safe controls using Definition 1 of the boat system used in the case study. The control set $\cset$ is given by $\{[\ctrl_1,\ctrl_2]\in\real^2| \  ||[\ctrl_1,\ctrl_2]||_2\leq 1\}$. The green-shaded region is the set of safe controls at the current state $\state$ and time $\tvar$, and the pink-shaded region is the set of controls that would lead the system into eventual violation of the state constraint. The location of the system relative to the boundary of the safe set, denoted using solid red line, is also illustrated in the top right subfigure. (Left) When the system is in the interior of the safe  set, any admissible controls are permitted/safe. (Right) When the system is on boundary or outside the safe set, only admissible controls that lead to small decrease in the safety value are permitted. }
    \label{fig:set_safe_ctrls}
\end{figure}

Intuitively, when the system is not at risk of exiting the safe set $\safeset$, the system is allowed to take any admissible control, and the system will remain within $\safeset$. Since $\frac{\partial\vfsafe}{\partial \tvar} + \frac{\partial\vfsafe}{\partial \state}^\top \dyn(\state, \ctrl)$ is the total derivative of $\vfsafe$ with respect to $\tvar$ along a state trajectory resulting from the applied control $\ctrl$, we can see that as we take $\gamma\rightarrow 0$, $\csetsafe(\state,\tvar)$ consists of only controls that instantaneously keep the safety value constant, when the system is on the boundary of $\safeset$. Though there is no admissible control that can render the system safe as soon as it exits $\safeset$, we define $\csetsafe(\state,\tvar)$ to be identical to the previous case. By doing so, we limit the degree to which the state constraint is violated and potential consequences, when the system finds itself outside of the safe set $\safeset$. We show an illustration of the set of safe controls based on Definition. \ref{def:set_of_safe_controls} of the system used in the case study in Fig. \ref{fig:set_safe_ctrls}. Recall that in order for the value function to be a viscosity solution of the HJB-PDE \eqref{eq:hjb_pde_perf}, the set-value map $\csetsafe(\state,\tvar)$ is required to be lower hemicontinuous. We now show \eqref{eq:set_of_safe_controls_hj} is lower hemicontinuous. 

\begin{proposition}
    Suppose $\vfsafe(\state,\tvar)$ is continuously differentiable in $\state$ and $\tvar$. The set-value map defined in \eqref{eq:set_of_safe_controls_hj} is lower hemicontinuous in $\state$ and $\tvar$. 
\end{proposition}

\begin{proof}

\textit{Case 1: $\vfsafe(\state,\tvar) > 0$.} In this case, $(\state,\tvar)$ is in the interior of the safe set $\safeset$, which we denote using $\interior{\safeset} =  \{(\state,\tvar)\in\sset\times\thor|\vfsafe(\state,\tvar)>0\}$. Take open set $\mathcal{A}\subset \cset$ such that $\mathcal{A}\cap\csetsafe(\state,\tvar)\neq \emptyset$. Since $\interior{\safeset}$ is open, $\exists \epsilon>0$ such that $\forall (\state', \tvar')\in\ball((\state,\tvar),\epsilon)$, we have $(\state', \tvar')\in\interior{\safeset}$. Then it follows that $\csetsafe(\state',\tvar') = \cset$, and $\csetsafe(\state',\tvar')\cap\mathcal{A} = \mathcal{A}\neq \emptyset$. Therefore, $\csetsafe(\state,\tvar)$ is lower hemicontinuous $\forall (\state,\tvar)\in\interior{\safeset}$. 

\textit{Case 2: $\vfsafe(\state,\tvar) < 0$.} In this case, $(\state,\tvar)$ is in the complement of the safe set $\safeset$, which we denote using $\overline{\safeset} = \{(\state,\tvar)\in\sset\times\thor|\vfsafe(\state,\tvar)< 0\}$. Take open set $\mathcal{A}\subset \cset$ such that $\mathcal{A}\cap\csetsafe(\state,\tvar)\neq \emptyset$. Since $\mathcal{A}$ is open and $\mathcal{A}\cap\csetsafe(\state,\tvar)\neq \emptyset$, $\exists \ctrl_0\in \mathcal{A}\cap\csetsafe(\state,\tvar)$ such that $-\gamma < \frac{\partial\vfsafe}{\partial \tvar} + \frac{\partial\vfsafe}{\partial \state}^\top \dyn(\state, \ctrl_0) < 0$.  Let $\nbhd = \ball((\state,\tvar),\delta)\subset \overline{\safeset}$ be an open ball centered at $(\state,\tvar)$. Take $(\state',\tvar')\in\nbhd$. Then for $\delta_1\in\real, \delta_2\in\real^{\sdim}, \delta_3\in\real^{\sdim}, \ \text{and} \ \delta_4\in\real^{\sdim\times\cdim} $, we have the following
\begin{equation}
    \begin{aligned}
        &\frac{\partial\vfsafe}{\partial \tvar}|_{(\state',\tvar')} + \frac{\partial\vfsafe}{\partial \state}|_{(\state',\tvar')}^\top \dyn(\state', \ctrl_0)\\
        &= \frac{\partial\vfsafe}{\partial \tvar}|_{(\state',\tvar')} + \frac{\partial\vfsafe}{\partial \state}|_{(\state',\tvar')}^\top\bigl[\dyn_1(\state')+\dyn_2(\state')\ctrl_0\bigr] \\
        &= \bigl[\frac{\partial\vfsafe}{\partial \tvar}|_{(\state,\tvar)} + \delta_1\bigr] + \bigl[\frac{\partial\vfsafe}{\partial \state}|_{(\state,\tvar)} + \delta_2 \bigr]^\top \biggl[\dyn_1(\state) \\
        & \qquad \qquad \qquad \qquad \qquad +\delta_3 + \bigl[\dyn_2(\state)+\delta_4\bigr]\ctrl_0\biggr] \\
        &= \bigl[ \frac{\partial\vfsafe}{\partial \tvar}|_{(\state,\tvar)} + \frac{\partial\vfsafe}{\partial \state}|_{(\state,\tvar)}^\top \dyn(\state,\ctrl_0)\bigr] \\
        &+ \delta_1 + \frac{\partial\vfsafe}{\partial \state}|_{(\state,\tvar)}^\top \delta_3 + \frac{\partial\vfsafe}{\partial \state}|_{(\state,\tvar)}^\top\delta_4\ctrl_0 + \delta_2^\top\dyn_2(\state)\ctrl_0 \\
        &+ \delta_2^\top\delta_4\ctrl_0 + \delta_2^\top\dyn_1(\state)+\delta_2^\top\delta_3
    \end{aligned}
\end{equation}

Using triangle inequality and definition of dot product, we have the following, where $||\cdot||$ denotes the Euclidean norm for a vector and the spectral norm for a matrix, and $|\cdot|$ denotes the absolute value of a real number.  

\begin{subequations}
    \begin{align}
    \begin{split}
        &\delta_1 + \frac{\partial\vfsafe}{\partial \state}|_{(\state,\tvar)}^\top \delta_3 + \frac{\partial\vfsafe}{\partial \state}|_{(\state,\tvar)}^\top\delta_4\ctrl_0 + \delta_2^\top\dyn_2(\state)\ctrl_0 \\
        &\qquad \qquad + \delta_2^\top\delta_4\ctrl_0 + \delta_2^\top\dyn_1(\state)+\delta_2^\top\delta_3
    \end{split} \\
    \begin{split}\label{eq:alt_safe_set_case_1_condition}
        &\leq |\delta_1| + ||\frac{\partial\vfsafe}{\partial \state}|_{(\state,\tvar)}|| \cdot ||\delta_3|| + ||\frac{\partial\vfsafe}{\partial \state}|_{(\state,\tvar)}|| \cdot ||\delta_4|| \\
        & \cdot ||\ctrl_0|| + ||\delta_2|| \cdot ||\dyn_2(\state)|| \cdot ||\ctrl_0|| + ||\delta_2|| \cdot ||\delta_4||  \\
        & \cdot ||\ctrl_0|| + ||\delta_2|| \cdot ||\dyn_1(\state)|| + ||\delta_2|| \cdot ||\delta_3||
    \end{split}
    \end{align}
\end{subequations}

Since $\vfsafe(\state,\tvar)$ is continuously differentiable in $\state$ and $\tvar$, and $\dyn_1(\state)$ as well as $\dyn_2(\state)$ are continuous in $\state$, we can choose $\delta$ such that $\forall (\state',\tvar')\in\nbhd=\ball((\state,\tvar),\delta)$ we have $\eqref{eq:alt_safe_set_case_1_condition} \leq \min\{\big\rvert\frac{\partial\vfsafe}{\partial \tvar}|_{(\state,\tvar)} + \frac{\partial\vfsafe}{\partial \state}|_{(\state,\tvar)}^\top \dyn(\state,\ctrl_0)\big\rvert, \big\rvert-\gamma-\frac{\partial\vfsafe}{\partial \tvar}|_{(\state,\tvar)} - \frac{\partial\vfsafe}{\partial \state}|_{(\state,\tvar)}^\top \dyn(\state,\ctrl_0)\big\rvert\}$, or equivalently $-\gamma \leq \frac{\partial\vfsafe}{\partial \tvar}|_{(\state',\tvar')} + \frac{\partial\vfsafe}{\partial \state}|_{(\state',\tvar')}^\top \dyn(\state', \ctrl_0) \leq 0$. We have shown that $\forall (\state',\tvar') \in \nbhd,  \ctrl_0\in\csetsafe(\state',\tvar')$, and as a result $\csetsafe(\state',\tvar')\cap\mathcal{A}\neq\emptyset$. Therefore, $\csetsafe(\state,\tvar)$ is lower hemicontinuous $\forall (\state,\tvar)\in\overline{\safeset}$. 

\textit{Case 3: $\vfsafe(\state,\tvar) = 0$.} In this case, $(\state,\tvar)$ is on the boundary of the safe set $\safeset$, which we denote using $\partial \safeset = \{(\state,\tvar)\in\sset\times\thor|\vfsafe(\state,\tvar) = 0\}$. Take open set $\mathcal{A}\subset \cset$ such that $\mathcal{A}\cap\csetsafe(\state,\tvar)\neq \emptyset$. We select $\epsilon>0$ and construct $\epsilon-$neighborhood around $(\state,\tvar)$, $\nbhd = \ball((\state,\tvar),\epsilon)$ such that $\forall (\state',\tvar')\in\nbhd\cap\overline{\interior{\safeset}}$, we have $\csetsafe(\state',\tvar')\cap\mathcal{A}\neq\emptyset$, using the argument presented above in Case 2. Note that $\forall (\state',\tvar')\in\nbhd\cap\interior{\safeset}$, we have $\csetsafe(\state',\tvar')=\cset$ and hence $\csetsafe(\state',\tvar')\cap\mathcal{A}=\mathcal{A}\neq\emptyset$. Since $(\nbhd\cap\overline{\interior{\safeset}}) \cup (\nbhd\cap\interior{\safeset})=\nbhd$, we have show that $\forall (\state',\tvar') \in \nbhd,  \csetsafe(\state',\tvar')\cap\mathcal{A}\neq\emptyset$, and $\csetsafe(\state,\tvar)$ is lower hemicontinuous $\forall (\state,\tvar)\in\partial\safeset$. 

We have exhausted all the cases, and therefore $\csetsafe(\state,\tvar)$ is lower hemicontinuous in $\state$ and $\tvar$. 
\end{proof}

\subsection{Synthesizing the Cooptimization Controller}
After obtaining the value function $\vfunc(\state,\tvar)$, we synthesize the closed-loop controller as follows 
\begin{equation}\label{eq:ctrl_synthesis}
    \controller^*(\state,\tvar) = \argmin_{\ctrl\in\csetsafe^*(\state,\tvar)}\{\frac{\partial\vfunc}{\partial\state}^\top \dyn(\state,\ctrl)+ r(\state,\ctrl)\}
\end{equation}

It is important to note that the controller synthesis problem \eqref{eq:ctrl_synthesis} is convex under the assumption that the running cost $r(\state,\ctrl)$ is convex in $\ctrl$ and the dynamics $\dyn(\state,\ctrl)$ are control-affine. Then it is clear that $\frac{\partial\vfunc}{\partial\state}^\top \dyn(\state,\ctrl)+ r(\state,\ctrl)$ is convex in $\ctrl$. Using the set of safe controls proposed in \eqref{eq:set_of_safe_controls_hj}, for $(\state,\tvar)$ such that $\vfsafe(\state,\tvar) > 0$, $\csetsafe(\state,\tvar)$ is the entire control space $\cset$, which is a convex set. On the other hand, for $(\state,\tvar)$ such that $\vfsafe(\state,\tvar) \leq 0$, $\csetsafe(\state,\tvar) = \{\ctrl\in\cset|\frac{\partial\vfsafe}{\partial \tvar} + \frac{\partial\vfsafe}{\partial \state}^\top \dyn(\state, \ctrl) = 0 \}$, the intersection of a hyperplane and a convex set, is also convex. Therefore, \eqref{eq:ctrl_synthesis} is an optimization problem with a convex objective and a convex constraint for any state $\state\in\sset$ and time $\tvar\in[\tinit,\tend)$. Very often $r(\state,\ctrl)$ depends quadratically on $\ctrl$ (e.g., to minimize the control energy), and for common choices of the control space $\cset$, such as hypercubes or Euclidean norm balls, \eqref{eq:ctrl_synthesis} is a quadratic program (QP) or a quadratically-constrained quadratic program (QCQP), both of which can be solved efficiently and reliably online.

\section{Case Study}
Since our method is ultimately solving a state-constrained optimal
control problem using dynamic programming, it is most similar to \cite{altarovici13}. 
To better compare against \cite{altarovici13}, we implement the numerical 
example from the paper with some minor modifications. The 2D system 
has the following dynamics $[\dot{x_1}, \dot{x_2}]^\top = [\ctrl_1 + 2 - 0.5 \state_2^2, \ctrl_2]$, with the control space $\cset = \{[\ctrl_1, \ctrl_2]\in\real^2| \ ||[\ctrl_1,\ctrl_2]||_2 \leq 1\}$.

The rectangular arena is given by $[-3,2]\times[-2,2]$, and the states outside of this arena are considered unsafe. There are two additional obstacles situated within the 
arena. The arena and the obstacle configuration are shown in Fig \ref{fig:boat2d_traj_comp}. 

The objective of the system is to minimize its distance to the goal location $[1.5,0]^\top$, giving rise to the cost functional $J(\state, \tvar, \ctrlseq) = \int_0^2 \sqrt{(x_1 - 1.5)^2 + x_2^2} \, d\tdummy$, while maintaining safety
(i.e. not running into the obstacles), over a time horizon of two seconds. We obtain the safety value function $\vfsafe(\state,\tvar)$ and the value function  $\vfunc(\state,\tvar)$ using LevelSetToolbox\cite{mitchell07lst} and HelperOC \cite{helperOC} using a grid size of $70\times70$. 

We evaluate our method and the baselines by synthesizing closed-loop control signals from 100 random initial states, and we focus primarily on 1) rollout success rate: the percent of trajectories that are safe over the entire time horizon, 2) rollout cost: the cost functional $\cost$ evaluated with the resulting control signals, and 3) offline and online computation time. 

The baselines we considered are Baseline 1) solving the state-constrained optimal control problem directly \cite{altarovici13}, Baseline 2) converting the state constraint \eqref{eq:state_constraint} into an obstacle penalty and solving the problem using Model Predictive Path Integral Control (MPPI) \cite{williams18mppi}, Baseline 3) performing safety filtering \cite{borquez23filtering} on the output of Baseline 2, and Baseline 4) solving the state-constrained optimal control problem in a receding horizon fashion (MPC).
The results are compiled in TABLE . \ref{table:boat2d_rollout_metrics}.

\begin{table}[h!]
\begin{center}
\caption{Comparison of metrics for our method and the baselines}
\label{table:boat2d_rollout_metrics}
\resizebox{9cm}{!}{
\begin{tabular}{|c|c|c|c|c|c|}
\hline
Method                                                                                                    & \begin{tabular}[c]{@{}c@{}}Our Method\\\end{tabular} & \begin{tabular}[c]{@{}c@{}}State-Constrained\\ Method\cite{altarovici13}\\(Baseline 1)\end{tabular} & \begin{tabular}[c]{@{}c@{}}MPPI\\ (Baseline 2)\end{tabular}    & \begin{tabular}[c]{@{}c@{}}MPPI+\\ Filtering \\ (Baseline 3)\end{tabular} & \begin{tabular}[c]{@{}c@{}}MPC\\ (Baseline 4)\end{tabular}     \\ \hline
\begin{tabular}[c]{@{}c@{}}Rollout success rate\end{tabular}                                           & 100\%                                                                  & 96\%                                                               & 19\%    & 100\%                                                     & 11\%    \\ \hline
\begin{tabular}[c]{@{}c@{}}\% of trajectories with \\ higher cost  compared \\to our method \end{tabular} & -                                                                      & 96.88\%                                                            & 84.21\% & 92\%                                                      & 72.73\% \\ \hline
\begin{tabular}[c]{@{}c@{}}Mean \% higher cost\\ compared to our method \end{tabular}                  & -                                                                      & 4.48\%                                                             & 14.96\% & 23.17\%                                                   & 2.13\%  \\ \hline
\begin{tabular}[c]{@{}c@{}}Offline computation time\end{tabular}                                      & 21 mins                                                                & 390 mins                                                           & -       & 3s                                                        & -       \\ \hline
\begin{tabular}[c]{@{}c@{}}Online computation time\end{tabular}                                       & 0.0015s                                                                & 0.0015s                                                            & 0.1s    & 0.1s                                                      & 0.6s    \\ \hline
\end{tabular}
}
\end{center}
\end{table}

We first analyze the rollout success rate, the metric indicative of the methods' ability to satisfy the safety requirement. In theory, Baseline 1 guarantees the satisfaction of the state constraint over the entire time horizon. However, Baseline 1 fails to achieved 100\% rollout success rate due to numerical inaccuracies that arise from the discretization of the state space. Baseline 2 and 4 performs poorly in this metric primarily due to the highly non-convex state constraint (disjoint obstacles in this case). On the other hand, our method and Baseline 3 are able to achieve 100\% rollout success rate. 

\begin{figure}[t!]
    \centering
    \includegraphics[width=1\linewidth]{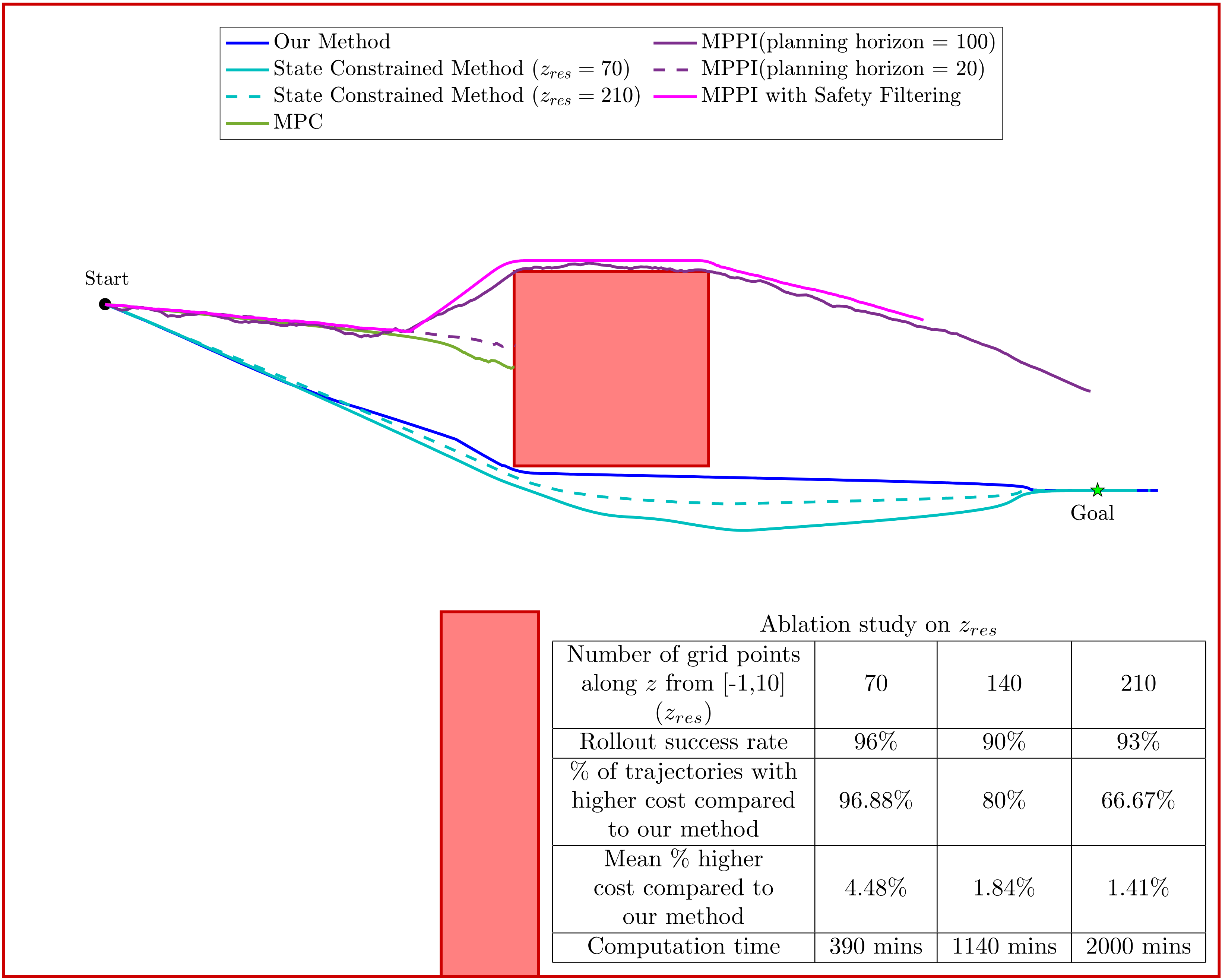}
    \caption{Trajectories from initial state $[-2.58, 0.77]^\top$.
    Costs incurred are (i) Our method: 3.51, (ii) State constrained method, $(z_{res} = 70)$: 3.62, (iii) State constrained method, $(z_{res} = 210)$: 3.52, (iv) MPPI (horizon = 20): Violates safety constraint (v) MPPI (horizon = 100): 4.55, (vi) MPPI with safety filtering: 4.95, (vii) MPC: Violates safety constraint }
    \label{fig:boat2d_traj_comp}
\end{figure}
 
In terms of the rollout cost, our method consistently outperforms Baseline 2 and 3 mostly due the fact that MPPI, in finite data regime, is only able to find locally optimal solution. Similarly, our method outperforms Baseline 4, because the non-convex optimization used in Baseline 4 is not solved to global optimum. Perhaps surprisingly, our method consistently outperforms Baseline 1. Though Baseline 1 and our method are computed using the same numerical tool, Baseline 1 is more severely affected by the discretization of the state space. Note that Baseline 1 augments its state space with an auxiliary state $z$ that is used to determine the actual value of the state. The discretization of the auxiliary state $z$ has a significant effect on the quality of the synthesized control signals, and we will demonstrate the effect using an ablation study on the number of grid points $z_{res}$ used in $z$'s dimension. The result of this ablation study is compiled in the bottom right table of Fig. \ref{fig:boat2d_traj_comp}. The performance of Baseline 1, in terms of trajectory cost, improves as the number of grid points in $z$ increases. However, the improvement of performance comes with a negative consequence of significant increase in offline computation time.    

We now examine the computation time. Compared to other methods, Baseline 1 and our method require the most offline computation, given the fact that the value functions are computed using dynamic programming on a grid \cite{mitchell07lst}. It is worthwhile to point out that our method requires significantly less offline computation, primarily due to the fact that our method does not require an auxiliary state as in \cite{altarovici13}, and the reduction in computation scales with the discretization of the auxiliary state. Though our method requires solving two optimal control problems, it is typically much faster than \cite{altarovici13} as the number of grid points used for the auxiliary state is much larger than 2. Baseline 3 requires some minimal offline computation for the safety value function. On the other hand, online methods Baseline 2 and Baseline 4 do not require any offline computation. For online computation time, Baseline 1 and our method outperform the rest of the baselines, as both methods solve quadratic programs, for which we use fast and reliable solver Gurobi \cite{gurobi}, for control synthesis online. 

We demonstrate the qualitative behaviors of the methods by showing the state trajectories, obtained using the synthesized closed-loop control signals over the entire time horizon, starting from a particular initial state in Fig. \ref{fig:boat2d_traj_comp}. The trajectory from our method is quite similar to that of Baseline 1, though the trajectory from Baseline 1 is slightly suboptimal for the aforementioned reasons. Baseline 2 and 3 unsurprisingly enter into a local minimum early on and are never able to recover. Baseline 4 fails to be safe as the corresponding optimization problem does not return the optimal solution satisfying the state constraint.

\section{Conclusion}
In this work, we proposed a method to synthesize controllers that cooptimize safety and performance for autonomous systems by formulating the problem as a control-constrained optimal control problem. We also show that the value function of the optimal control problem is a viscosity solution to a certain HJB-PDE. Although our method is shown to provide safety guarantee for the system and outperform other methods in terms of performance, our method has several drawbacks. First, while the theory is general, our method does not scale to high-dimensional systems. In the future, we will look into computing the value function using deep learning techniques \cite{bansal_deepreach, fisac_hj_rl}. Furthermore, to synthesize controllers, our method assumes that the safety value function $\vfsafe(\state,\tvar)$ and the value function $\vfunc(\state,\tvar)$ are differentiable everywhere, which is typically not the case. We will explore overcoming this challenge using a smooth overapproximation of the value functions \cite{borquez23filtering}.

\section*{Acknowledgement}
We would like to thank Sanat Mulay for his insights and help in the proof of Proposition 1.

\bibliographystyle{plain}
\bibliography{citations} 

\end{document}